\documentclass[conference]{IEEEtran}
\usepackage{amsfonts}
\usepackage{amssymb}
\usepackage{amsthm}
\usepackage{float}
\usepackage{graphicx}
\usepackage{amssymb,amsthm} 
\usepackage{mathtools}
\usepackage{amsmath, amsfonts,epsfig, multirow, floatflt}
\usepackage{epsfig,graphics,graphicx,latexsym,amsfonts,amssymb,amsmath,verbatim}
\usepackage{mathrsfs}
\usepackage{subfigure}
\usepackage{cite}
\usepackage{color}
\usepackage{soul}
\usepackage{multirow}
\usepackage{multicol}
\usepackage{array}
\usepackage{algorithm}
\usepackage{algpseudocode}
\usepackage{setspace}
\usepackage{booktabs}
\usepackage{threeparttable}
\usepackage{makecell}
\usepackage{epstopdf}
\usepackage{tabularx}
\usepackage{stfloats}
\usepackage{geometry}
\usepackage{cases}
\usepackage{amsmath,bm} 
\newcounter{MYalgorithmic}

\graphicspath{{figure/}}
\usepackage{caption}
\ifCLASSINFOpdf
\else
\fi
\newtagform{brackets}{(}{)}
\usetagform{brackets}
\usepackage{graphicx}
\usepackage{float} 
\usepackage{subfigure}

\begin{document}

\title{Cooperative V2X for High Definition Map Transmission Based on Vehicle Mobility}
\author{\IEEEauthorblockN
{Fangfei Wang, Dong Guan, Long Zhao, \textit{Member, IEEE}, and Kan Zheng, \textit{Senior Member, IEEE} }
\IEEEauthorblockA{\\Intelligent Computing and Communication Lab,\\Key Laboratory of Universal Wireless Communication, Ministry of Education, \\Beijing University of Posts and Telecommunications,\\ Beijing, China, 100876 \\
}}
\maketitle

\begin{abstract}
High-definition (HD) map transmission is considered as a key technology for automatic driving, which enables vehicles to obtain the precise road and surrounding environment information for further localization and navigation. Guaranteeing the huge requirement of HD map data, the objective of this paper is to reduce the power consumption of vehicular networks. By leveraging the mobile rule of vehicles, a collaborative vehicle to everything (V2X) transmission scheme is proposed for the HD map transmission. Numerical results indicate that the proposed scheme can satisfy the transmission rate requirement of HD map with low power consumption.
\end{abstract}

\vspace{-0.4em}
\section{Introduction}
Vehicular networks is a typical application of 5G ultra-reliable and low-latency communication (URLLC). It is also an important mean to realize automatic driving in intelligent transportation systems (ITS) \cite{1}. Automatic driving technology is helpful to avoid traffic accidents and reduce traffic congestion \cite{2,3}. In order to realize automatic driving technology, the accuracy of positioning and controlling for automatic driving vehicles (ADVs) should be up to the centimeter level. Therefore, a series of sensor information should be provided for the ADVs in order to support the centimeter accuracy \cite{4}.\par

A rapid lane detection algorithm was proposed based on machine vision,  which is accurate and robust under different conditions, such as lane line missing and obstacle appearance in the track \cite{5}. However, many vision-based research challenges have not yet been solved, such as the low-definition image clarity and poor visibility in rainy, hazy weather and night conditions \cite{6}. Since vision-based automatic driving probably has safety risks in the situations mentioned above, high-definition (HD) map emerges to support automatic driving. The road information contained by HD map has enough precision to help ADVs identify the road signs with a centimeter accuracy. On the other hand, HD map also contains real-time traffic information, such as the state information of running cars, pedestrians and cyclists, which will be helpful to avoid accidents in critical situations with fast response times \cite{7}. Moreover, HD map-based vehicle localization and predictive cruise control have been studied in \cite{8,9}. However, how to transmit the HD map by wireless network is still an open problem to our best knowledge. Since HD map contains more road and traffic related information than traditional maps, the data volume is huge for the network to delivery. Therefore, it is necessary to design some transmission schemes to support high transmission rate with low network cost and latency. \par
\parskip=0pt
Vehicle-to-infrastructure (V2I) makes vehicles connect to the neworks via roadside units (RSUs) and supports high-speed short-range communications \cite{10}. Therefore, a transmission scheme based on V2I communication has been proposed through jointly optimizing the traffic flow rate and the power consumption of the network \cite{11}. On the other hand, in addition to transmitting HD map by RSU, two cars driving in the opposite direction contain the HD map information needed by each other. Therefore, vehicle-to-vehicle (V2V) communication can also be employed for HD map transmission. Besides, V2V communications have shorter communication distance than V2I, which can reduce the path loss and transmission delay \cite{12,13}. 
\par

Motivated by this, this paper studies the HD map transmission for automatic driving. Considering both the power efficiency and communication efficiency, a collaborative V2X transmission scheme is proposed in order to achieve high-speed HD map transmission with low power cost. The proposed scheme combines both V2I and V2V communications and adaptively allocates the power between the RSU and vehicle. On the other hand, a more realistic estimation expression of the transmission rate is adopted in order to reflect the influence of the decoding error probability. The simulation results indicate that the proposed transmission scheme can reduce power consumption while guaranteeing the transmission rate requirement of HD map.\par
The remainder of this paper is organized as follows. Section \uppercase\expandafter{\romannumeral2} describes the system model and formulates the HD map transmission problem with power consumption minimization. A cooperative transmission scheme is proposed based on mobile vehicle rule in Section \uppercase\expandafter{\romannumeral3}. The simulation results and analysis are given in Section \uppercase\expandafter{\romannumeral4}. Section \uppercase\expandafter{\romannumeral5} concludes this paper.
\par 
\begin{figure*}[!t]
	\centering
	\subfigure[V2I transmission only]{
		\label{Fig.sub.1}
		\includegraphics[width=0.31\textwidth]{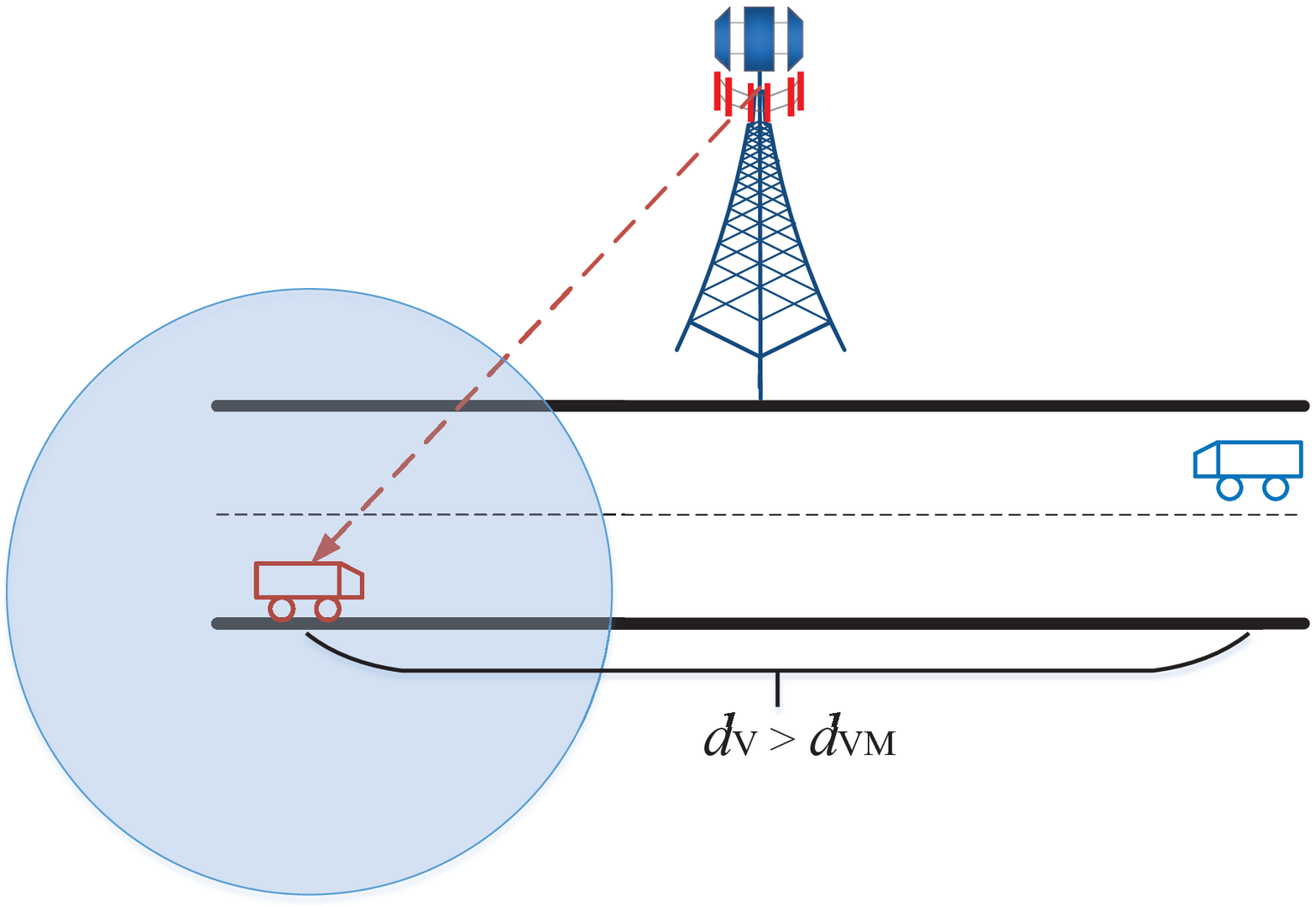}}
	\subfigure[Cooperative V2X transmission]{
		\label{Fig.sub.2}
		\includegraphics[width=0.31\textwidth]{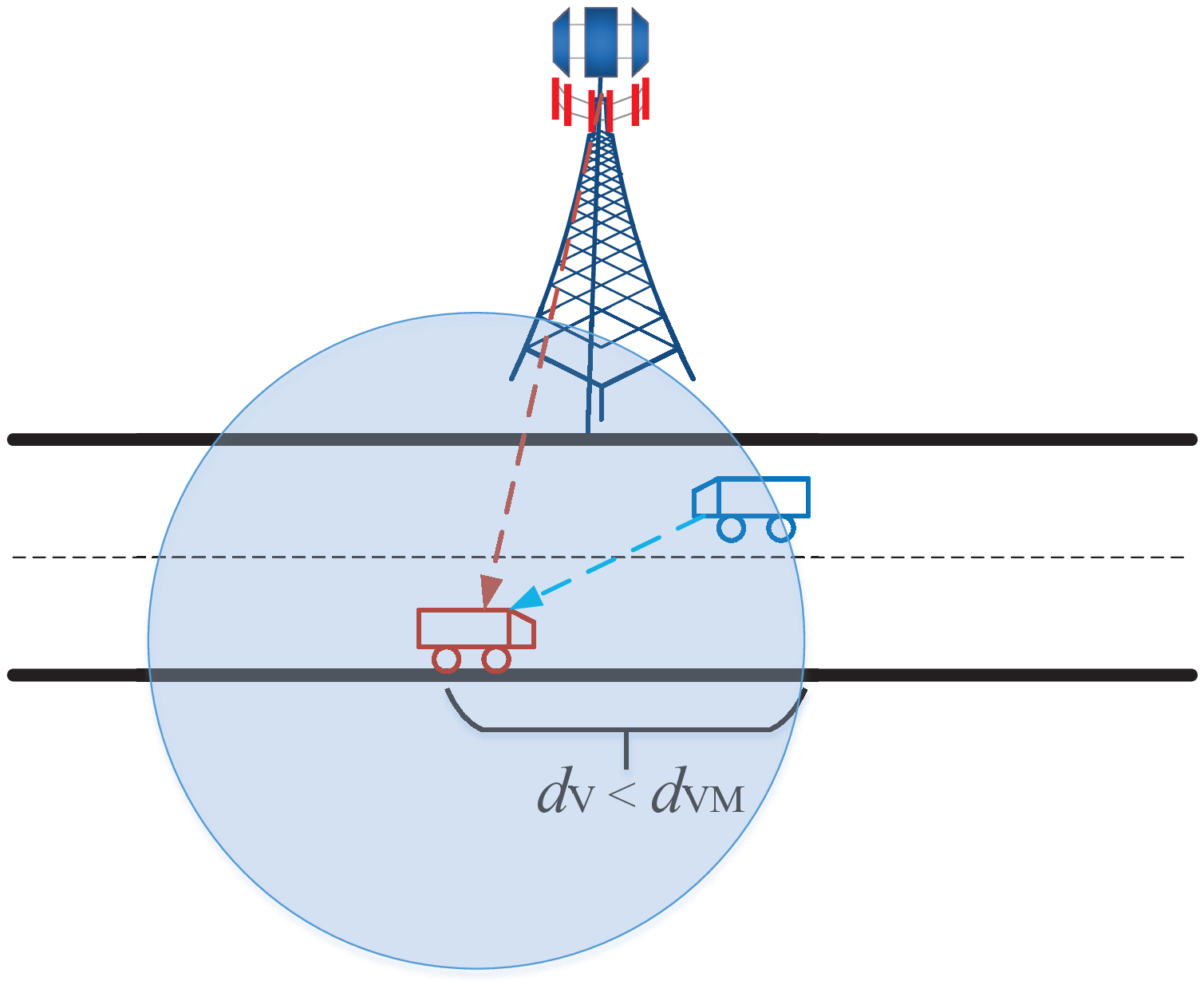}}
	\subfigure[V2V transmission only]{
		\label{Fig.sub.3}
		\includegraphics[width=0.31\textwidth]{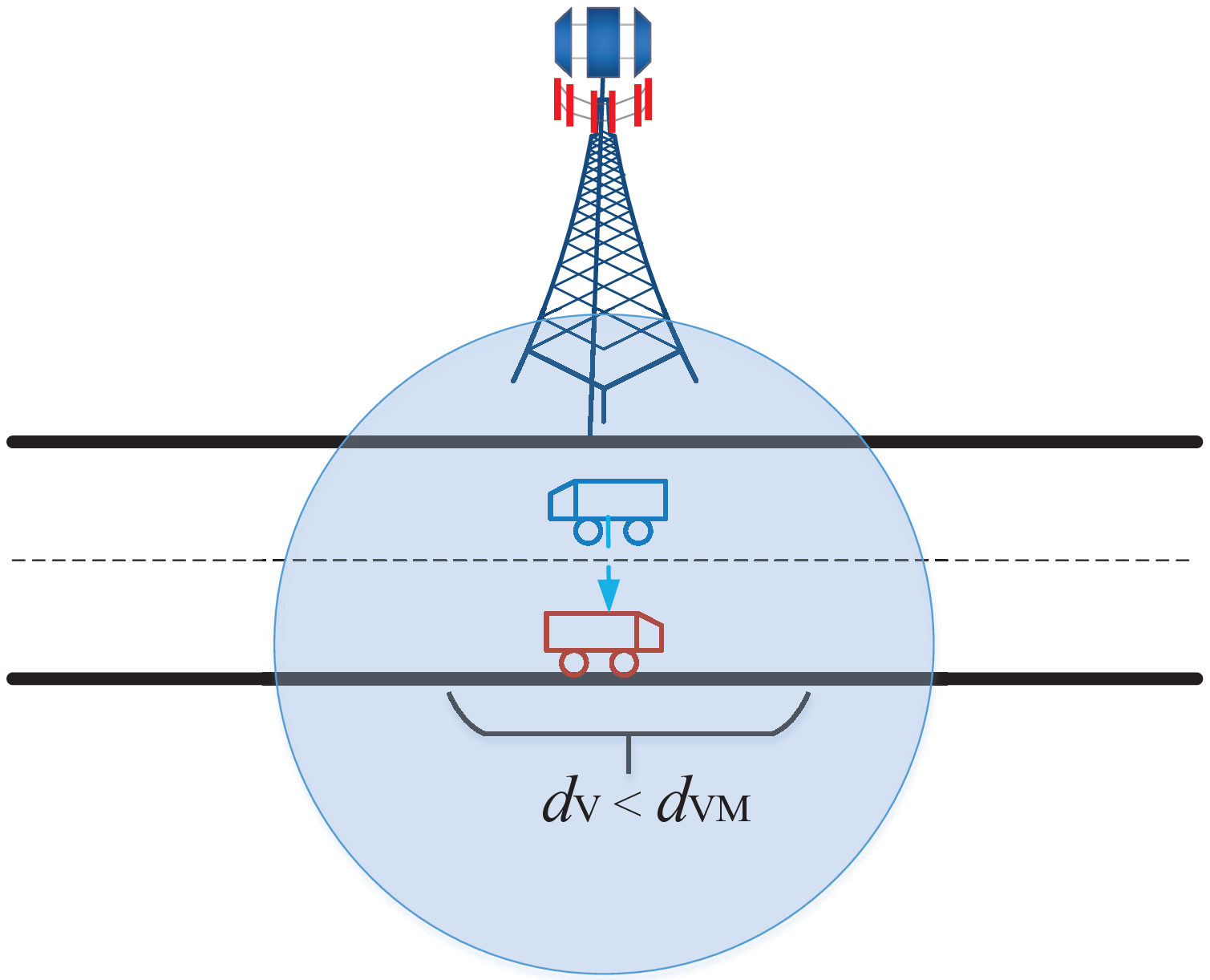}}
		\vspace{-6pt}
	\caption{Illustration of cooperative transmission of high definition map.}
	\vspace{-1.3em}
	\label{Fig.main}
\end{figure*}
\section{System Model and Problem Formulation}

\subsection{System Model}

As shown in Fig. 1, two vehicles driving in the opposite directions on an urban road, where the red one is considered as the targeted vehicle of the HD map requirement, the blue one has stored the required HD map information in the forward direction of the red vehicle. 
The different frequency bands with bandwidths $B_{\rm V}$ and $B_{\rm R}$ are employed by V2V and V2I in order to avoid interference. Then, different parts of HD map information can be simultaneously transmitted by the blue vehicle and the RSU for the targeted vehicle. Moreover, it is assumed that the HD map data volume per meter $q$ bits/m follows truncated Gaussian distribution, i.e., $q \sim \mathbb{N}(\mu,{{\sigma} ^2},0,+ \infty )$, where $\mu$ and ${\sigma}$ are the mean and variance, and the data volume is positive. When the vehicle speed is $v$ m/s, the transmission rate of the networks at least is $vq$ bits/s.\par

We assume that ${d_{{\rm{V}}}}$ denotes the distance between the two vehicles and ${d_{{\rm{VM}}}}$ indicates the maximum V2V communication range. As shown in Fig. 1, three cases of transmission schemes are then discussed based on the position relationship of the two vehicles. In Fig. 1(a), only RSU transmits the HD map information to the targeted vehicle when ${d_{{\rm{V}}}}>{d_{{\rm{VM}}}}$. In Fig. 1(b), when ${d_{{\rm{V}}}}<{d_{{\rm{VM}}}}$, both RSU and vehicle transmit the HD map information to the targeted vehicle. In Fig. 1(c),  ${d_{{\rm{V}}}}$ becomes so small that the HD map transmission rate can be satisfied only by the V2V communication. Next, the channel models, transmission model and transmission rate are given in details.
\subsubsection{Channel Model}
Denote $\phi_0 {\it{d}_{\rm{V}}^{ - \alpha }}$ and $\phi_0 {\it{d}_{\rm{R}}^{ - \alpha }}$ as the gains of downlink large-scale channels for V2V and V2I, respectively, where $\phi_0$ is a channel constant related to the antenna gain and carrier frequency,  $\alpha$ is the path-loss exponent, $d_{\rm R}$ represents the distance between the red vehicle and RSU. On the other hand, denote ${h}_{\rm V}$ and ${h}_{\rm R}$ as the fast fading coefficients of the V2V and V2I channels, respectively, and  ${h}_{\rm V}$, ${h}_{\rm R} \sim \mathbb{CN}(0,1)$. Then, the V2V and V2I channels can be expressed as ${g_{\rm V}} \rm{ = } \phi_0 {\it{d}_{\rm V}^{ - \alpha }}{\it h}_{\rm V}$ and ${g_{\rm R}} \rm{ = } \phi_0 {\it{d}_{\rm{R}}^{ - \alpha }} \it{h}_{\rm R}$, respectively.
\subsubsection{Transmission Model}
Assuming that the transmit power values of the blue vehicle and RSU are $p_{\rm V}$ and $p_{\rm R}$, and $s_{\rm V}$ with $\left| s_{\rm V} \right| \rm{ = } 1 $ and $s_{\rm R}$ with $\left| s_{\rm R} \right| \rm{ = } 1 $ represent the transmitted symbols from the blue vehicle and RSU, respectively. Then the transmission signals of the blue vehicle and RSU can be written respectively by
\vspace{-0.5em}
\begin{align}
\vspace{-0.3em}
&x_{\rm V} { = } \sqrt{p_{\rm{V}}}s_{\rm V},\\
&x_{\rm R} { = } \sqrt{p_{\rm{R}}}s_{\rm R},
\end{align}
and the received signals of the red vehicle are respectively given by
\vspace{-0.5em}
\begin{align}
\vspace{-0.5em}
\notag y_{\rm V} &= {{g}_{\rm V} \it {x_{\rm V}  }+n_{\rm V}}\\
 &={{{\phi_0 {\it{d}_{\rm V}^{ - \alpha }}{{h}_{\rm V}}\sqrt {p_{\rm{V}}} {s_{\rm V}}}  + n_{\rm V}}},\\
\notag y_{\rm R} &= {{g}_{\rm R} \it {x_{\rm R}  }+n_{\rm R}}\\
&={{{\phi_0 {\it{d}_{\rm R}^{ - \alpha }}{{h}_{\rm R}}\sqrt {p_{{\rm R}}} {s_{\rm R}}}  + n_{\rm R}}},
\end{align}
where the additive white Gaussian noises $n_{\rm V}\sim \mathbb{CN}(0,{B_{\rm V}N_0}), n_{\rm R} \sim \mathbb{CN}(0,{B_{\rm R}N_0})$. Therefore, the signal to noise ratios (SNRs) of the red vehicle can be expressed respectively as
\begin{align}
&\rho_{\rm V}= {\frac{{{\phi _0}d_{\rm{V}}^{ - \alpha }{h_{\rm{V}}}{p_{\rm{V}}}}}{{{\phi _1}{N_0}{B_{\rm{V}}}}}},\\
&\rho_{\rm R}= {\frac{{{\phi _0}d_{\rm{R}}^{ - \alpha }{h_{\rm{R}}}{p_{\rm{R}}}}}{{{\phi _1}{N_0}{B_{\rm{R}}}}}},
\end{align}
where $\phi _1$ is a SNR loss coefficient due to non-ideal channel state information at the transmitter.
\par
\subsubsection{Transmission rate}
Considering the decoding error probability, the transmission rates of vehicle and RSU can be expressed respectively as 
\begin{align}
{r_{\rm{V}}}\! =\! \frac{{{B_{\rm{V}}}}}{{\ln 2}}\!\left[\! {\ln \!\left(\! {1\!+\!\frac{{{\phi _0}d_{\rm{V}}^{\! -\! \alpha }{h_{\rm{V}}}{p_{\rm{V}}}}}{{{\phi _1}{N_0}{B_{\rm{V}}}}}} \!\right) \!\!-\!\! \sqrt {\frac{1}{{\tau {B_{\rm{V}}}}}} Q_{\rm{G}}^{ \!- \!1}\!\left(\! {\varepsilon _{\rm{V}}^{{\rm{C,V}}}}\! \right)}\! \right]\!\!,
\end{align}
\begin{align}
{r_{\rm{R}}}\! = \!\frac{{{B_{\rm{R}}}}}{{\ln 2}}\!\left[\! {\ln \!\left(\! {1 \!+\! \frac{{{\phi _0}d_{\rm{R}}^{\! -\! \alpha }{h_{\rm{R}}}{p_{\rm{R}}}}}{{{\phi _1}{N_0}{B_{\rm{R}}}}}}\! \right)\! \!-\!\! \sqrt {\frac{1}{{\tau {B_{\rm{R}}}}}} Q_{\rm{G}}^{ \!- \!1}\!\left(\! {\varepsilon _{\rm{R}}^{{\rm{C,R}}}} \!\right)}\! \right]\!\!,
\end{align}
where $\tau$ is the duration of transmission, $ {\varepsilon _{\rm{V}}^{{\rm{C,V}}}}$ and $ {\varepsilon _{\rm{R}}^{{\rm{C,R}}}} $　are decoding error probabilities in downlink of V2V channel and V2I channel, respectively, and $Q_{\rm{G}}^{ - 1}\left( x \right)$ denotes the inverse of the Gaussian Q-function \cite{14,15}.
\par

\subsection{Problem Formulation}

Based on the proposed collaborative V2X transmission scheme, the objective of this paper is to minimize the total power consumption of both the RSU and vehicles as well as ensure the transmission rate requirement of the HD map. Therefore, the optimization problem of this paper can be formulated as \par
\vspace{-1em}
\begin{align}
&\mathop {\min }\limits_{{p_{\rm{R}}},{p_{\rm{V}}}} \left\{ {{p_{\rm{R}}} + {p_{\rm{V}}}} \right\}\\
{\rm{s}}{\rm{.t}}&{\rm{.  0}} \le {p_{\rm{R}}} \le {p_{{\rm{RM}}}},\\
&{\rm{       0}} \le {p_{\rm{V}}} \le {p_{{\rm{VM}}}},\\
&{\rm{       P}}\left\{ {{r_{\rm{R}}} + {r_{\rm{V}}} < qv} \right\} \le \delta ,\\
&{\rm{       0}} \le {d_{\rm{V}}} \le {d_{{\rm{VM}}}},
\end{align}
where (10) and (11) are the transmit power constraints of vehicle and RSU, ${p_{{\rm{RM}}}}$ and ${p_{{\rm{VM}}}}$ denote the maximum transmit power values of RSU and vehicle, respectively, (12) describes the outage probability requirement of HD map transmission and $\delta$ is the maximum violation probability.
\par

\section{Cooperative power allocation for HD Map Transmission}
In order to solve the formulated problem described in (9)-(13), we need to transform the transmission rate outage constraint (12) into a constraint  of $p_{\rm{R}}$ or $p_{\rm{V}}$ by considering the distribution of the HD map volume $q$. Since the expression of tansmission rate in (7) or (8) contains a constant term, (12) has different expressions in different situations, i.e.
\vspace{-0.5em}
\begin{align}
	&{{\rm{P}}\left\{ {{r_{\rm{R}}} < qv} \right\} \le \delta ,}\\
	&{{\rm{ P}}\left\{ {{r_{\rm{V}}} < qv} \right\} \le \delta ,{\rm{ }}}\\
	&{{\rm{ P}}\left\{ {{r_{\rm{R}}} + {r_{\rm{v}}} < qv} \right\} \le \delta .}
\end{align} 
We will solve the problem in (9)-(13) under the conditions (14)-(16), respectively, i.e., V2I transmission only, V2V transmission only, and cooperative V2X transmission. Then, we can obtain three suboptimal power allocation schemes. Finally, the optimal power allocation can be obtained by comprehensively considering the sum transmit power values of both RSU and vehicle.

\subsection{Suboptimal Power Allocation}
\newtheoremstyle{mythm}{}{}{\normalfont}{}{ \bfseries \it}{\normalfont:}{.5em}{}
\theoremstyle{mythm}
\newtheorem{theorem}{\quad \it Proposition} 
\newtheorem{Commen}{\quad \it Comment } 
\renewenvironment{proof}{{\it Proof \normalfont:}}{ \hfill $\blacksquare$ }
\subsubsection{V2I transmission only}
We have the following results. 
\begin{theorem}
When the HD map is transmitted by V2I only, i.e., $p_{\rm{V}}=0$, the power allocation of the vehicle and the RSU is given by 
\begin{align}
{\psi _1} = \left( {{p_{\rm{V}}},{p_{\rm{R}}}} \right),
\end{align}
where
\begin{align}
\notag {p_{\rm{R}}} =& \!\left\{\! {\frac{{{\phi _1}{N_0}{B_{\rm{R}}}}}{{{\phi _0}d_{\rm{R}}^{ \!-\! \alpha }{h_{\rm{R}}}}}\!\!\left[\! {\exp \!\!\left(\! {\frac{{v\!\!\left(\! {\sigma {\Phi ^{ \!-\!1}}\!\left(\! {1 \!-\! \delta \!\left[\! {1\!-\! \Phi \!\left(\! { \!-\!{\textstyle{\mu  \over \sigma }}} \!\right)\!} \right]\!} \right)\! \!+\! \mu } \!\right)\!\ln 2}}{{{B_{\rm{R}}}}}} \right.} \right.} \right.\\
&\left. {\left. {\left. { \!+\! \sqrt {\frac{1}{{\tau {B_{\rm{R}}}}}} Q_{\rm{G}}^{ \!-\! 1}\!\left(\! {\varepsilon _{\rm{R}}^{{\rm{C,R}}}} \!\right)}\! \right)\! \!-\! 1} \!\right]\!} \right\}_0^{{p_{{\rm{RM}}}}}\!\!\!\!,
\end{align}
where $\Phi \left( x \right)$ is the cumulative distribution function of standard normal distribution, and $\{x\}_a^b=\min \{b, \max \{a,x\}\}$.
\end{theorem}
\begin{proof}
To obtain the power allocation under the conditions (14),
we need to transform the transmission rate constraint (12). The probability density function of $q$ can be given by
\begin{align}
\!\!f\left( q \right) &\!=\! \frac{{{\textstyle{1 \over {{\sigma }}}}{f_N}\left( {{\textstyle{{q \!-\! \mu } \over {{\sigma }}}}} \right)}}{{1 \!-\! \Phi \left( { \!-\! {\textstyle{\mu  \over {{\sigma }}}}} \right)}}\!\!=\!\!\frac{{\exp \left( { \!-\! \frac{1}{2}{{\left( {{\textstyle{{q\!-\!\mu } \over {{\sigma }}}}} \right)}^2}} \right)}}{{\sqrt {2\pi } {\sigma}\left[ {1 \!-\! \Phi \left( {\!-\!{\textstyle{\mu  \over {{\sigma}}}}} \right)} \right]}},
\end{align}
where ${f_N}\left( x \right)$ is the probability density function of standard normal distribution. Then, the left side of inequality (12) can be transformed into
\begin{align}
\notag &{\rm{P}}\!\left\{\!\! {\frac{{{B_{\rm{R}}}}}{{\ln 2}}\!\!\left[\! {\ln \!\!\left(\!\! {1 \!+\! \frac{{{\phi _0}d_{\rm{R}}^{ \!-\! \alpha }{h_{\rm{R}}}{P_{\rm{R}}}}}{{{\phi _1}{N_0}{B_{\rm{R}}}}}}\! \!\right)\! \!\!-\!\! \sqrt {\frac{1}{{\tau {B_{\rm{R}}}}}} Q_{\rm{G}}^{ \!-\! 1}\!\left(\! {\varepsilon _{\rm{R}}^{{\rm{C,R}}}} \!\right)}\! \right]\!\! <\! qv} \!\right\}\\
\!=&{\rm{P}}\!\left\{\!\! {q \!>\!\! \frac{{{B_{\rm{R}}}}}{{v\ln 2}}\!\!\left[\! {\ln \!\!\left(\!\! {1 \!+\! \frac{{{\phi _0}d_{\rm{R}}^{ \!-\! \alpha }{h_{\rm{R}}}{P_{\rm{R}}}}}{{{\phi _1}{N_0}{B_{\rm{R}}}}}} \!\!\right)\! \!-\!\! \sqrt {\frac{1}{{\tau {B_{\rm{R}}}}}} Q_{\rm{G}}^{ \!-\! 1}\!\left(\! {\varepsilon _{\rm{R}}^{{\rm{C,R}}}}\! \right)\!} \!\right]\!\!} \right\}\!\!\!.
\end{align}\par
In order to simplify the expression, we define the constant 
\begin{align}
\kappa {\rm{ = }} \frac{{{B_{\rm{R}}}}}{{v\ln 2}}\!\!\left[\! {\ln \!\!\left(\!\! {1 \!+\! \frac{{{\phi _0}d_{\rm{R}}^{ \!-\! \alpha }{h_{\rm{R}}}{P_{\rm{R}}}}}{{{\phi _1}{N_0}{B_{\rm{R}}}}}} \!\!\right)\! \!-\!\! \sqrt {\frac{1}{{\tau {B_{\rm{R}}}}}} Q_{\rm{G}}^{ \!-\! 1}\!\left(\! {\varepsilon _{\rm{R}}^{{\rm{C,R}}}}\! \right)\!} \!\right]\!\!.
\end{align}
Then, based on (19) and (21), the outage probability in (20) can be expressed as
\begin{align}
\notag \int_\kappa ^{\! +\! \infty }\!\! {f\!\!\left( q \right)\!{\rm d}q} =& \int_\kappa ^{ \!+\! \infty } \!\!{\frac{{\exp \left( { \!-\! \frac{1}{2}{{\!\left( \!{{\textstyle{{q\! -\! \mu } \over {{\sigma }}}}}\! \right)\!}^2}} \right)}}{{\sqrt {2\pi } {\sigma }\left[ {1 \!-\! \Phi \!\left(\! { - {\textstyle{\mu  \over {{\sigma }}}}}\! \right)\!} \right]}}{\rm d}q} \\
=& \frac{{1 \!-\! \Phi \!\left(\! {{\textstyle{{\kappa  \!-\! \mu } \over {{\sigma}}}}} \!\right)\!}}{{1 \!-\! \Phi \!\left(\! { \!-\! {\textstyle{\mu  \over {{\sigma }}}}} \!\right)\!}},
\end{align}
and therefore (12) can be rewritten as
\begin{align}
\frac{{1 \!-\! \Phi \!\left(\! {{\textstyle{{\kappa  \!-\! \mu } \over {{\sigma }}}}} \!\right)\!}}{{1 \!-\! \Phi \!\left(\! { \!-\! {\textstyle{\mu  \over {{\sigma }}}}} \!\right)\!}} \le \delta .
\end{align}
Further, taking into account equation (21), constraint (12) can be transformed into
\begin{align}
\notag {p_{\rm{R}}} \!\ge\! \!\frac{{{\phi _1}{N_0}{B_{\rm{R}}}}}{{{\phi _0}d_{\rm{R}}^{ - \alpha }{h_{\rm{R}}}}}\!\!&\left[ \!{\exp\! \left(\!\! {\frac{{v\!\left(\! {\sigma {\Phi ^{ \!-\! 1}}\!\!\left(\! {1 \!-\! \delta\! \left[\! {1\! - \!\Phi \!\left(\! {\! - \!{\textstyle{\mu  \over \sigma }}}\! \right)} \!\right]\!} \right) \!+\! \mu } \!\right)\!\ln 2}}{{{B_{\rm{R}}}}}} \right.} \right.\\
&\left. {\left. {\! +\! \sqrt {\frac{1}{{\tau {B_{\rm{R}}}}}} Q_{\rm{G}}^{ \!-\! 1}\!\left(\! {\varepsilon _{\rm{R}}^{{\rm{C,R}}}}\! \right)} \!\right)\!\! - \!1} \right].\end{align} 
\par
Considering constraints (10) and (11), the power allocation under V2I transmission only can be given by (17). 
\end{proof}

\subsubsection{V2V transmission only}
We have the following results. 
\begin{theorem}
When the HD map is transmitted by V2V only, i.e., $p_{\rm R}=0$, the power allocation of the vehicle and the RSU is given by 	
\begin{align}
{\psi _2} = \left( {{p_{\rm{V}}},{p_{\rm{R}}}} \right),
\end{align}
where  
\begin{align}
\notag {p_{\rm{V}}} =& \!\left\{\! {\frac{{{\phi _1}{N_0}{B_{\rm{V}}}}}{{{\phi _0}d_{\rm{V}}^{ \!-\! \alpha }{h_{\rm{V}}}}}\!\!\left[\! {\exp \!\!\left(\! {\frac{{v\!\!\left(\! {\sigma {\Phi ^{ \!-\!1}}\!\left(\! {1 \!-\! \delta \!\left[\! {1\!-\! \Phi \!\left(\! { \!-\!{\textstyle{\mu  \over \sigma }}} \!\right)\!} \right]\!} \right)\! \!+\! \mu } \!\right)\!\ln 2}}{{{B_{\rm{V}}}}}} \right.} \right.} \right.\\
&\left. {\left. {\left. { \!+\! \sqrt {\frac{1}{{\tau {B_{\rm{V}}}}}} Q_{\rm{G}}^{ \!-\! 1}\!\left(\! {\varepsilon _{\rm{V}}^{{\rm{C,V}}}} \!\right)}\! \right)\! \!-\! 1} \!\right]\!} \right\}_0^{{p_{{\rm{VM}}}}}\!\!\!\!.
\end{align}
\end{theorem}

\begin{proof}
The proof is similar to that of {\it Proposition} 1 and therefore is ignored.
\end{proof}

\subsubsection{Cooperative V2X transmission}We have the following results. 
\begin{theorem}
When the HD map is cooperatively transmitted by V2V and V2I, the power allocation of vehicle and RSU is given by	
\begin{align}	
{\psi _3}= \left( {{{p}_{\rm{V}}},{p_{\rm{R}}}} \right),
\end{align}
where
\begin{align}
&{p_{\rm{V}}}=\!{{\frac{{{\phi _1}{N_0}{B_{\rm{V}}}}}{{{\phi _0}d_{\rm{V}}^{ \!-\! \alpha }{h_{\rm{V}}}}}\!\!\!\left[ \!\!{{{\left(\! {\frac{{d_{\rm{R}}^{ -\! \alpha }{h_{\rm{R}}}}}{{d_{\rm{V}}^{ -\! \alpha }{h_{\rm{V}}}{{\rm e}^{{\chi}}}}}} \!\right)}^{ \!-\! \frac{{{B_{\rm{R}}}}}{{{B_{\rm{V}}} \!+\! {B_{\rm{V}}}}}}} \!\!\!\!\!\!-\! 1}\! \right]_0^{{p_{{\rm{VM}}}}}} }\! \!\!\!\!,\\
&{p_{\rm{R}}}=\!\! {{\frac{{{\phi _1}{N_0}{B_{\rm{R}}}}}{{{\phi _0}d_{\rm{R}}^{ -\! \alpha }{h_{\rm{R}}}}}\!\!\left[\! {{e^{{\chi}}}{{\left(\! {\frac{{d_{\rm{R}}^{ -\! \alpha }{h_{\rm{R}}}}}{{d_{\rm{V}}^{ \!-\! \alpha }{h_{\rm{V}}}{{\rm e}^{{\chi }}}}}} \!\right)}^{\frac{{{B_{\rm{V}}}}}{{{B_{\rm{V}}} \!+\! {B_{\rm{R}}}}}}} \!\!\!\!\!\!-\! 1} \!\right]_0^{{p_{{\rm{RM}}}}}}\!\!}\!\!.
\end{align}
\end{theorem}

\begin{proof}
To obtain the power allocation of cooperative HD map transmission, we use the method in {\it Proposition} 1 to transform (12) into
\begin{align}
{p_{\rm{R}}} \!\ge\! g\!\left(\! {{p_{\rm{V}}}} \!\right),
\end{align}
where
\begin{align}
\!\!\!\!g\!\left(\! {{p_{\rm{V}}}} \!\right) \!=\! \frac{{{\phi _1}{N_0}{B_{\rm{R}}}}}{{{\phi _0}d_{\rm{R}}^{\! -\! \alpha }{h_{\rm{R}}}}}\!\!\left[\!\! {{{\rm e}^{{\chi}}}{{\!\left(\! {1 \!+\! \frac{{{\phi _0}d_{\rm{V}}^{ \!-\! \alpha }{h_{\rm{V}}}{p_{\rm{V}}}}}{{{\phi _1}{N_0}{B_{\rm{v}}}}}} \!\right)\!}^{\! -\! \frac{{{B_{\rm{V}}}}}{{{B_{\rm{R}}}}}}} \!\!\!-\! 1} \!\right]\!\!,
\end{align}
with
\begin{align}
\!\!\!\!&\chi \!=\! \frac{{ \{\sigma {\Phi ^{ \!-\! 1}}\left[ {1 \!-\! \delta \left( {1 \!-\! \Phi \left( { \!-\! {\textstyle{\mu  \over \sigma }}} \right)} \right)} \right] \!+\! \mu {\rm{\!-\! }}{\textstyle{\vartheta \over v}}\}v\ln 2}}{{{B_{\rm{r}}}}},\\
\!\!\!\!&\vartheta \!=\! \! -\! \frac{{{B_{\rm{r}}}}}{{\ln 2}}\!\sqrt {\frac{1}{{\tau {B_{\rm{r}}}}}} Q_{\rm{G}}^{ \!-\! 1}\left(\! {\varepsilon _{\rm{r}}^{{\rm{c,r}}}}\! \right) \!\!-\!\! \frac{{{B_{\rm{v}}}}}{{\ln 2}}\!\sqrt {\frac{1}{{\tau {B_{\rm{v}}}}}} Q_{\rm{G}}^{\! -\! 1}\left(\! {\varepsilon _{\rm{v}}^{{\rm{c,v}}}} \!\right)\!\!.
\end{align}
\par
Without consideration of (13), the partial Lagrange function of problem in (9) is given by
\begin{align}
\!\!\notag {\mathcal L}\!\left(\! {{p_{\rm{R}}},{p_{\rm{V}}},\lambda }\! \right)\!&\!=\!{p_{\rm{R}}} \!+\! {p_{\rm{V}}} \!-\! {\lambda _1}{p_{\rm{R}}} \!-\! {\lambda _2}{p_{\rm{V}}} \!+\! {\lambda _3}\!\left(\! {{p_{\rm{R}}} \!-\! {p_{{\rm{RM}}}}} \!\right)\\
&\!+\!{\lambda _4}\left( {{p_{\rm{V}}} \!-\! {p_{{\rm{VM}}}}} \right) \!-\! {\lambda _5}\left[ {{p_{\rm{R}}} \!-\! g\left( {{p_{\rm{V}}}} \right)} \!\right]\!\!,
\end{align}
where ${\bm \lambda} {\rm{ = }}\left\{ {{\lambda _i} \ge 0,i = 1,..., 5} \right\}$ is the Lagrange multiplier vector. We can prove that the first order derivative of $g\left( {{p_{\rm{V}}}} \right)$ with respect to $p_{\rm{V}}$ is less than zero and the second-order derivative of $g\left( {{p_{\rm{V}}}} \right)$ with respect to $p_{\rm{V}}$ is greater than zero, therefore the considered optimization problem is a convex problem. Then, the optimization problem can be solved based on Karush-Kuhn-Tucker (KKT) conditions. With ${{ p}_{\rm{R}}} \ne 0$and ${{ p}_{\rm{V}}} \ne 0$, the KKT conditions are given by

\begin{align}
&0 < {{ p}_{\rm{R}}} \le {p_{{\rm{RM}}}},\\
&0 < {{ p}_{\rm{V}}} \le {p_{{\rm{VM}}}},\\
&{{ p}_{\rm{R}}} \!-\! g\left( {{p_{\rm{V}}}} \right)\!\ge\! 0,\\
&{{ \lambda }_i} \ge 0,i = 1,...,5,\\
&{{ \lambda }_1}{\rm{ = }}{{ \lambda }_2}{\rm{ = }}0,\\
&{{ \lambda }_3}\left( {{{ p}_{\rm{R}}} - {p_{{\rm{RM}}}}} \right){\rm{ = }}0,\\
&{{ \lambda }_4}\left( {{{ p}_{\rm{V}}} - {p_{{\rm{VM}}}}} \right){\rm{ = }}0,\\
&{{ \lambda }_5}\left({ p}_{\rm{R}}- g\left({p_{\rm{V}}} \right)\right)=0,\\  
&\nabla {\mathcal L}=0.
\end{align}

Then, substituting (34) into (43) results in
\begin{align}
&{\rm{1}} \!+\! {{ \lambda } _3} \!-\! {{ \lambda } _5}{\rm{ = }}0,\\
&{\rm{1}} \!+\! {{ \lambda } _4} \!-\! {{ \lambda } _5}\frac{{d_{\rm{V}}^{\!-\! \alpha }{h_{\rm{V}}}}}{{d_{\rm{R}}^{ \!-\! \alpha }{h_{\rm{R}}}}}{e^{{{\rm{A}}_3}}}{\!\left(\! {1 \!+\! \frac{{{\phi _0}d_{\rm{V}}^{ \!-\! \alpha }{h_{\rm{V}}}{{ P}_{\rm{V}}}}}{{{\phi _1}{N_0}{B_{\rm{V}}}}}} \!\right)^{ \!\!\!-\! \frac{{{B_{\rm{V}}}}}{{{B_{\rm{R}}}}} \!-\! 1}} \!\!\!\!\!\!\!=\! 0.
\end{align}
Finally, according to KKT conditions (35)-(45), we can obtain the power allocation given in (27).
\end{proof}
\subsection{Optimal Power Allocation}
By comprehensively considering the three cases in Sec. III-A, the optimal power allocation can be expressed by 
\begin{align}
{\psi ^*}\left( {p_{\rm{V}}^*,p_{\rm{R}}^*} \right){\rm{ = argmin}}&\left\{ {{p_{{\rm{V}\it{i}}}} + {p_{{\rm{R}\it{i}}}},\left( {{p_{{\rm{V}\it{i}}}},{p_{{\rm{R}\it{i}}}}} \right)} \right.\\
&\left. { \in \left\{ {{\psi _1},{\psi _2},{\psi _3}} \right\}} \right\}.
\end{align}

\begin{table}[!b]
	\centering
	\scriptsize
	\renewcommand{\arraystretch}{1.4}
	\vspace{-1.5em}
	\caption{Simulation Parameters}
	\vspace{0pt}
	\label{parameters}
	\begin{tabular}{|p{4.7cm}<{\centering}|p{2.1cm}<{\centering}|}\hline   
		\footnotesize \textbf {Parameter} & \footnotesize \textbf {Value}\\\hline    
		Channel constant ($\phi_0$) &$10^{-3}$ \\\hline
		path-loss exponent ($\alpha$) &3\\\hline
		SNR loss coefficient ($\phi_1$) &1.5\\\hline
		Duration of transmission ($\tau$) &$10^{-3}$ s\\\hline
		Noise power spectrum density (${N_0}$) &-174 dBm/Hz\\\hline
		RSU and vehicle bandwidth ($B_{\rm{R}}$, $B_{\rm{V}}$) & 1 MHz, 0.5 MHz\\\hline
		Maximum transmit power of RSU and vehicle ($p_{\rm{RM}}$,$p_{\rm{VM}}$) &40 dBm, 36 dBm \\\hline
		Communication range of V2V ($d_{\rm{VM}}$) &150 m\\\hline
		Decoding error probabilities ($ {\varepsilon _{\rm{V}}^{{\rm{C,V}}}}\!={\varepsilon _{\rm{R}}^{{\rm{C,R}}}} $) &$10^{-4}$\\\hline
		Mean and variance of HD map data volume ($\mu$, $\sigma$) &0.8 kbits/m, $100$\\\hline
		Maximum violation probability ($\delta$)	&$10^{-4}$ \\\hline
	\end{tabular}    
\end{table}

\section{Simulation Results and Analysis}
In this section, numerical results and analysis are provided to show the performance of the proposed cooperative V2X transmission scheme.

\subsection{Parameter Setup}
The simulation parameters are listed in Table \uppercase\expandafter{\romannumeral1}. The road length covered by the RSU is 432 m, and the distance between the RSU and the road is 250 m, the lane width is 3.5 m in Fig. 1. In order to simplify the simulation, we assume that the two vehicles have the same speed. 
\par

\subsection{Results and Analysis}


Fig. 2 shows the average transmit power of RSU versus different locations of the targeted vehicle. With the horizontal change of the targeted vehicle, both the distances $d_{\rm R} $, $d_{\rm V} $ first decrease and then gradually increase. This change rule gives rise to the same trend of total transmit power. When the location of the targeted vehicle between 144 m - 288 m, we have $d_{\rm V} < d_{\rm VM}$; therefore V2V communication exists and the average transmit power obviously decreases. It indicates that the proposed collaborative transmission can significantly reduce the power consumption.


Fig. 3 and Fig. 4 illustrate the average transmission rate  and average power allocation versus vehicle speed, respectively. The transmission rate requirement is related to both vehicle speed and HD map data volume per meter. Due to the power limitation of RSU, the V2I transmission only can not satisfy the transmission rate requirement when the vehicle speed exceeds 22 m/s. On the other hand, limited by the communication range of V2V, V2V transmission only can not meet the transmission rate requirement of the HD map. However, under the proposed cooperative V2X transmission, the transmission rate requirement can be satisfied when the vehicle speed is under 30 m/s, which approaches the maximum speed limitation of vehicle. Moreover, Fig. 4 shows that the average transmit power under the cooperative V2X transmission is significantly reduced than that under V2I transmission only. Therefore, the proposed cooperative transmission can achieve the HD map transmission with low power consumption.
\par

\begin{figure}[!t]
	\centering
	\includegraphics[width=0.44\textwidth]{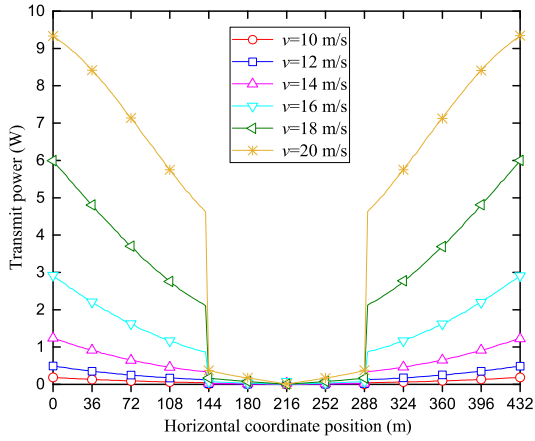}
	\caption{Average transmit power v.s. the location of the targeted vehicle.}	
	\vspace{-8pt}
\end{figure}

\section{CONCLUSION}
This paper studied the power efficient transmission of HD map, which is significant for automatic driving. In order to reduce power consumption while guaranteeing the transmission rate requirement, a cooperative V2V/V2I transmission was proposed for HD map transmission. To realize the cooperative transmission scheme, the power allocation at both RSU and vehicle are given through solving three sub optimization problems. Finally, the simulation results indicated that the proposed scheme can significantly reduce the total power consumption compared to the V2I transmission scheme while meeting the transmission rate requirement of HD map.\par
\begin{figure}[!t]
	\centering
	\includegraphics[width=0.44\textwidth]{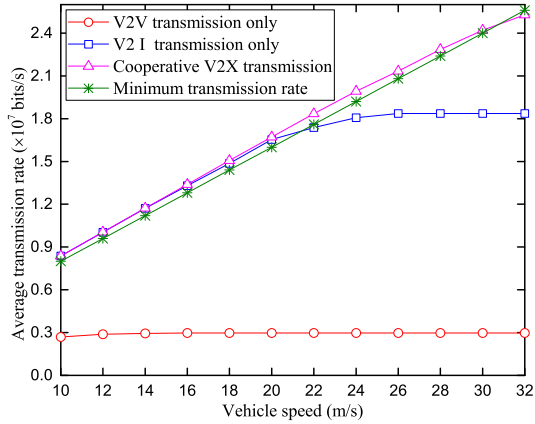}
	\caption{
		Average transmission rate v.s. vehicle speed.}
	\vspace{-12pt}	
\end{figure}

\section*{Acknowledgement}
This work was supported by the China Natural Science Funding under Grant 61731004.
\begin{figure}[!t]
	\centering
	\includegraphics[width=0.44\textwidth]{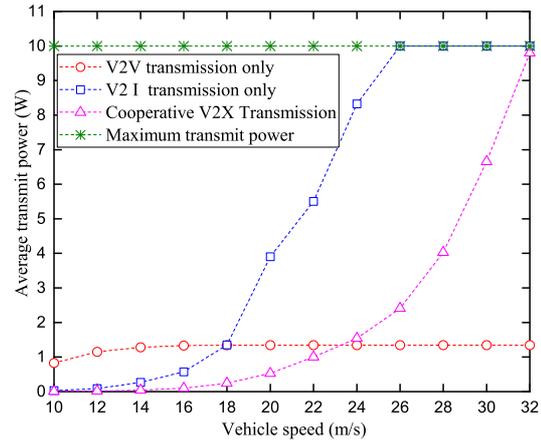}
	\caption{
		Average transmit power v.s. vehicle speed.}	
	\vspace{-10pt}
\end{figure}


\begin{thebibliography}{99}
\bibitem{1}
H. Zhang, N. Liu, X. Chu, K. Long, A. Aghvami, and V. Leung, ``Network slicing based 5G and future mobile networks: Mobility, resource management, and challenges," \textit{IEEE Commun. Mag.}, vol. 55, no. 8, pp. 138-145, Aug. 2017

\bibitem{2}
Q. Yuan, H. Zhou, J. Li, Z. Liu, F. Yang, and X. S. Shen, ``Toward efficient content delivery for automated driving services: An edge computing solution," \textit{IEEE Netw.}, vol. 32, no. 1, pp. 80-86, Jan./Feb.  2018.

\bibitem{3}
K. Zheng, Q. Zheng, H. Yang, L. Zhao, L. Hou, and P. Chatzimisios, ``Reliable and efficient autonomous driving: The need for heterogeneous vehicular networks," \textit{IEEE Commun. Mag.}, vol. 53, no. 12, pp. 72-79, Dec. 2015.
 
\bibitem{4} 
S. W. Kim, B. Qin, Z. J. Chong, X. Shen, W. Liu, M. Ang, E. Frazzoli, and D. Rus, ``Multivehicle cooperative driving using cooperative perception: Design and experimental validation," \textit{IEEE Trans. Intell. Transport. Syst.}, vol. 16, no. 2, pp. 663-
680, Apr. 2015.

\bibitem{5}
K. D. T. Nguyen, L. D. Nguyen, S. H. Le, T. Van Le, and V. Nguyen, ``Vision-based driverless car in the condition of limited computing resource: Perspectives from a student competition," \textit{in the Proceedings of  2017 21st Asia Pacific Symposium on Intelligent and Evolutionary Systems (IES)}, Vietnam, Hanoi, pp. 67-72, Nov. 2017.

\bibitem{6}
A. Bar Hillel, R. Lerner, D. Levi, and G. Raz, ``Recent progress in road and lane detection: A survey," \textit{Mach. Vision Appl.}, vol. 23, no. 3, pp. 727-745, Apr. 2014.

\bibitem{7}
H. G. Seif, and X. Hu, ``Autonomous driving in the iCity−HD maps as a key challenge of the automotive industry," \textit{Engineering}, vol. 2,  no. 2, pp. 159-162, Jun. 2016.

\bibitem{8}
S. Zheng, and J. Wang, ``High definition map-based vehicle localization for highly automated driving: Geometric analysis," \textit{in Proceedings of the 2017 International Conference on Localization and GNSS (ICL-GNSS)}, Nottingham, UK, pp. 1-8, Jun. 2017. 

\bibitem{9}
H. Chu, L. Guo, B. Gao, H. Chen, N. Bian, and J. Zhou, ``Predictive cruise control using high-definition map and real vehicle implementation," \textit{IEEE Trans. Veh. Technol.}, vol. 67, no. 12, pp. 11377-11389, Dec. 2018.

\bibitem{10}
K. Zheng, L. Hou, H. Meng, Q. Zheng, N. Lu, and L. Lei, ``Soft-defined heterogeneous vehicular network: Architecture and challenges," \textit{IEEE Netw. Mag.}, vol. 30, no. 4, pp. 72-80, Jul./Aug. 2016.

\bibitem{11}
L. Zhao, F. Wang, K. Zheng, and T. Riihonen, ``Joint Optimization of communication and traffic efficiency in vehicular networks," \textit{IEEE Trans. Veh. Technol.}, 2019, online publication.

\bibitem{12}
J. Mei, K. Zheng, L. Zhao, Y. Teng, and X. Wang, ``A latency and reliability guaranteed resource allocation scheme for LTE V2V communication systems," \textit{IEEE Trans. Wireless Commun.}, vol. 17, no. 6, pp. 3850-3860, Jun. 2018.

\bibitem{13}
C. Chen, H. Jinna, T. Qiu, M. Atiquzzaman, and Z. Ren, ``CVCG: Cooperative V2V-aided transmission scheme based on coalitional game for popular content distribution in vehicular Ad-hoc networks," \textit{IEEE Trans. Mobile Comput.}, 2018, online publication.


\bibitem{14}
W. Yang, G. Durisi, T. Koch, and Y. Polyanskiy, ``Quasi-static multiple-antenna fading channels at finite blocklength," \textit{IEEE Trans. Inf. Theory}, vol. 60, no. 7, pp. 4232-4264, Jul. 2014.

\bibitem{15}
C. Sun, C. She, C. Yang, T. Q. S. Quek, Y. Li, and B. Vucetic, ``Optimizing resource allocation in the short blocklength regime for ultra-reliable and low-latency communications," \textit{IEEE Trans. Wierl. Commun.}, vol. 18, no. 1, pp. 402-415, Jan. 2019.

\end{thebibliography}
\end{document}